\documentclass[conference]{IEEEtran}

\usepackage{hyperref}
\hypersetup{citecolor=blue,pdfborder=0 0 0,pdfstartview={FitH}}
\ifCLASSINFOpdf
 \usepackage[pdftex]{graphicx}
\else
\fi
\usepackage[top=0.54in, bottom=0.54in, left=0.732in, right=0.732in]{geometry}

\usepackage{caption}
\usepackage{subcaption}

\usepackage[cmex10]{amsmath}

\newtheorem{theorem}{Theorem}

\newtheorem{algorithm}{Algorithm}
\newtheorem{definition}{Definition}
\newtheorem{problem}{Problem}
\newtheorem{example}{Example}
\newenvironment{proof}[1][Proof]{\begin{trivlist}
\item[\hskip \labelsep {\bfseries #1}]}{\end{trivlist}}

\newcommand{\qed}{\nobreak \ifvmode \relax \else
      \ifdim\lastskip<1.5em \hskip-\lastskip
      \hskip1.5em plus0em minus0.5em \fi \nobreak
      \vrule height0.75em width0.5em depth0.25em\fi}
\begin{document}
%
\title{Algorithms and Throughput Analysis for MDS-Coded Switches}

\author{\IEEEauthorblockN{Rami Cohen and Yuval Cassuto}
\IEEEauthorblockA{Department of Electrical Engineering\\
Technion - Israel Institute of Technology\\
Technion City, Haifa 3200003, Israel\\
Email: rc@tx.technion.ac.il, ycassuto@ee.technion.ac.il}}


\maketitle
\maketitle

\begin{abstract} Network switches and routers need to serve packet writes and reads at rates that challenge the most advanced memory technologies. As a result, scaling the switching rates is commonly done by parallelizing the packet I/Os using multiple memory units. For improved read rates, packets can be coded with an [n,k] MDS code, thus giving more flexibility at read time to achieve higher utilization of the memory units. In the paper, we study the usage of [n,k] MDS codes in a switching environment. In particular, we study the algorithmic problem of maximizing the instantaneous read rate given a set of packet requests and the current layout of the coded packets in memory. The most interesting results from practical standpoint show how the complexity of reaching optimal read rate depends strongly on the writing policy of the coded packets.
\end{abstract}


%
\IEEEpeerreviewmaketitle

\section{Introduction}

Ever increasing demand for network bandwidth pressures switch and router vendors to scale their products at a fast pace. The most crucial component for throughput scaling is the memory sub-system that comprises the switching fabric. As it becomes extremely difficult or prohibitively costly to scale the read and write rates of memory units (MUs), an alternative solution of choice is to deploy multiple such units in parallel. Memory contention in switches occurs when multiple packets requested for read at a given instant happen to reside in the same memory unit. Because every memory unit can deliver a single chunk of data per time instant, this contention will result in a loss of switching throughput (while at the same time a different memory unit will be idle). Our ability to avoid such contention by clever packet placement is limited by the fact that the reading schedule of packets is not known upon arrival of the packets to the switch.

This issue has driven switch vendors to seek methods to reduce memory contention and thus increase the switching throughput. One particularly promising technique is to introduce {\em redundancy} to the packet-write path, such that upon read the switch controller will enjoy greater flexibility to read the requested packets from memory units not contended by other read requests. This redundancy is introduced in the form of {\em coding}, whereby additional coded chunks are calculated from the incoming packet and written along with it in the switch memory. The simplest scheme of coding applied to packets is {\em replication}, where the additional chunks written with the packet are simply copies of the same packet. The advantage of this scheme is in its simplicity, requiring only trivial encoding and decoding. A more advanced packet-coding scheme uses $[n,k]$ maximum distance separable (MDS) codes \cite{Huffman}. This coding scheme takes an input of $k$ packet chunks and encodes them into a codeword of $n$ chunks ($k \le n$), where \textit{any} $k$ chunks taken from the $n$ code chunks  can be used for the recovery of the original $k$ packet chunks. This maximal flexibility in retrieving the packet makes MDS codes very well suited to use in switch memories. Latency comparison between replication and MDS codes was pioneered by Huang et al. \cite{Huang}. It was shown that for $k=2$, the average latency for serving a packet decreases significantly when a certain scheduling model is used. This analysis was later extended by Shah et al. in \cite{Shah, Shah2}, where bounds on latency performance under multiple scheduling policies were investigated.

In this paper, we provide a model of a coded switch that considers the number of memory units in use and the code parameters. Then we put our focus on maximizing and analyzing the {\em throughput} of the switch. The notion of throughput we pursue here is the active MUs serving packets out of the MUs in the switch. Increasing the number of packets in the system (i.e., the load) gives more choice to the reader, and is thus expected to improve the throughput. For the probabilistic analysis of throughput we use a static distribution on the requested-packets' locations (governed by the write policy), and assume that the switch observes this static distribution in steady state. The problem of achieving maximal throughput is formulated in two equivalent ways as problems in set theory and in graph theory. These formulations allow us later to obtain insights, algorithms, and bounds for the problem. 

This paper is structured as follows. In Section \ref{sec:problem_formulation}, we provide complexity analysis of the throughput maximization problem, where in Section \ref{sec:prob_analysis} we discuss suboptimal algorithms and bounds on the optimal solution. We then provide a structured version of the problem in Section \ref{sec:special_cases} and show that it admits a polynomial-time solution. Finally, conclusions are given in Section \ref{sec:conclusions}.

\section{Problem Formulation and Complexity}
\label{sec:problem_formulation}

Consider a scenario in which multiple packets are stored in MUs, to be later forwarded by a switch. Let us assume that each packet consists of $k$ chunks, which are MDS-encoded into $n$ chunks. The $n$ chunks are then stored in $n$ MUs out of $N$ available ones ($k \le n \le N$), where overlapping is allowed (i.e., two or more packets may share one or more MUs). Out of many packets currently stored in the switch memory, a request arrives for $L$ packets, with the objective to read as many out of these packets in a single time instant. Recall that $k$ chunks out of the $n$ encoded ones are sufficient for recovering a packet. We know in which MUs these $L$ packets are stored, and wish to find methods for reading as many packets as possible simultaneously, with the constraint that each MU can be accessed \textit{only once} in the reading process. Let us denote by $L^{*}$ the maximal number of packets that can be read. We consider the following notion of throughput as a performance measure. 
\begin{definition}(Throughput)

The throughput $\rho$ of the system is defined as:
\begin{equation}
\rho = \frac{{L^{*}k}}{N}.
\end{equation}
\end{definition}
We name the problem of maximizing the throughput $\rho$ as the \textit{$[n,k]$-maximal throughput problem}, or nkMTP. An instance of the problem is illustrated in Figure \ref{fig:nkMTP_illus}, where data chunks of multiple packets are shown on top of MUs storing them. 
\begin{figure}[t]
\centering
\includegraphics[scale=0.58]{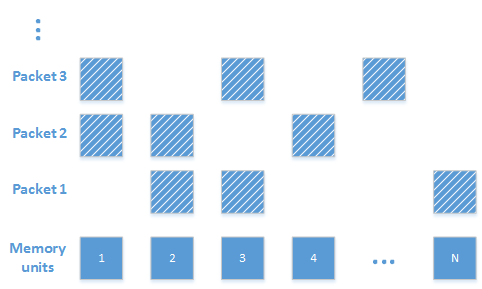}
\caption{Illustration of nkMTP. The patterned rectangles represent encoded data chunks ($n=3$).}
\label{fig:nkMTP_illus}
\end{figure}

The motivation for using MDS codes in this setting rather than simple replication is demonstrated in Figure \ref{fig:rep_mds_comparison}. Here we compare the use of a [$4$,$2$] MDS code for encoding $2$ chunks of a packet to $2$ uses of a [$2$,$1$] repetition code (i.e., 2-way replication), one use for each chunk. To read a packet when the repetition code is used, the reader needs to read one chunk from a specific subset of $2$ MUs out of the $4$ MUs storing the packet, and an additional chunk from the other subset of $2$ MUs. On the other hand, when a [$4$,$2$] MDS code is used, the reader can recover the packet by reading two chunks from \textit{any} $2$ MUs out of the $4$ MUs storing the packet chunks. The resulting curves in Figure \ref{fig:rep_mds_comparison} show that on average the [$4$,$2$] MDS code allows reading more simultaneous packets than the [$2$,$1$] repetition code. Note that both schemes write $4$ chunks per packet, and hence this advantage comes without increasing the write load.

\begin{figure}[t]
\centering
\includegraphics[scale=0.65]{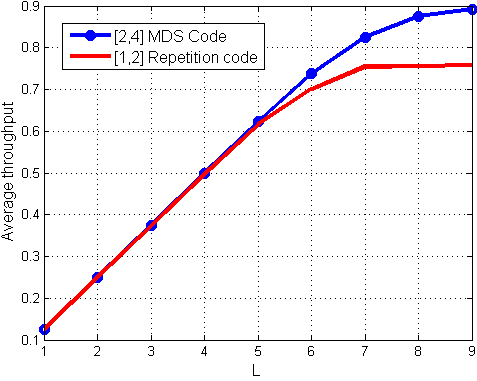}
\caption{Comparison between a [$4,2$] MDS code and two uses of the [$2,1$] repetition code ($N=16$).}
\label{fig:rep_mds_comparison}
\end{figure}

Evidently, the way packets are stored dictates how they should be read for optimal results. In particular, each write/read scheme combination may result in different number of packets that can be read. In this work, we will consider two write schemes that are easy to implement, and their optimally matched read schemes. For analysis purposes, we provide in the rest of this section two equivalent formulations of nkMTP. Consider the $N$ available MUs as the elements of the set $S = \left\{ {1,2,...,N} \right\}$. Each packet $i=1,2,...,L$ is stored in MUs indexed by a subset $S_i$ of $S$, where $\left| {{S_i}} \right| = n$ and the subsets may overlap. The set theory formulation of nkMTP is as follows.
\begin{problem}{(nkMTP, set theory formulation)}
\label{nkMTP_set}

\textbf{Input}: Set $S = \left\{ {1,2,...,N} \right\}$ and $L$ subsets of $S$, $S_i \subseteq S$, such that $\left| {{S_i}} \right| = n$.

\textbf{Output}: Subsets $S'_i \subseteq S_i$, $\left| S'_i \right| = k$, $S'_i \cap S'_j = \emptyset$ ($i \ne j$), such that the number of subsets is maximal.
\end{problem}

\begin{example}
\label{ex1}

$N=5, L=3, n=3$. The packets are stored in the MUs indexed by the sets ${S_1} = \left\{ {1,2,3} \right\}, {S_2} = \left\{ {2,4,5} \right\}, {S_3} = \left\{ {3,4,5} \right\}$. If $k=n=3$, we have that $L^{*}=1$ and the recovered packet can be either $1,2$ or $3$ since ${S_i} \cap {S_j} \ne \emptyset$ for all $i,j=1,2,3$. If $k=2$, a possible solution is ${S'_1 = \left\{ {1,2} \right\},S'_2 = \left\{ {4,5} \right\}}$ with $L^*=2$. Note that no more than $2$ packets can be read in this case, since ${L^*} \le \left\lfloor {{N}/{k}} \right\rfloor  = 2$. Finally, if $k=1$ all the packets can be read, and one possible solution is ${S'_1 = \left\{ 1 \right\},S'_2 = \left\{ 2 \right\},S'_3 = \left\{ 3 \right\}}$. 
\end{example} 

In addition to the set theory formulation, nkMTP can be formulated equivalently on a graph. Consider a \textit{bipartite} graph $G = \left( {{V_G},{E_G}} \right)$, where $V_G$ denotes the vertices of $G$ and $E_G$ denotes the edges of $G$. In addition, let us denote by $\deg(v)$ the \textit{degree} of a vertex $v \in V_G$. Since $G$ is bipartite, $V_G$ can be partitioned into two \textit{disjoint} sets of vertices, let us denote them by $X_G$ and $Y_G$. Thinking of $X_G$ as packets, and of $Y_G$ as MUs, vertex $i$ in $X_G$ will be connected to vertex $j$ in $Y_G$ if one of the encoded chunks of packet $i$ is stored in MU $j$. The resulting graph has the following properties: $\deg(x)=n$ ($\forall x \in X_G$), $\left| X_G \right| = L$, $\left| Y_G \right| = N$, and $\left| E_G \right| = nL$. nkMTP can be now formulated as follows.
\begin{problem}{(nkMTP, graph theory formulation)}
\label{nkMTP_graph}

\textbf{Input}: Graph representation $G$ of nkMTP.

\textbf{Output}: Subsets $X_H \subseteq X_G, Y_H \subseteq Y_G, E_H \subseteq E_G$ with the properties $\deg (x) = k,\deg (y) \in \left\{ {0,1} \right\}$ ($\forall x \in X_H,y \in Y_H$), such that ${\left| {X_H} \right|}$ is maximal.
\end{problem}

For instance, Example $\ref{ex1}$ is represented as the graph in Figure \ref{fig_ex1}. The output of Problem \ref{nkMTP_graph} is essentially a \textit{subgraph} $H$ of $G$, providing a solution to nkMTP with throughput $\rho  = {{\left| {{X_H}} \right|k}}/{N}$. 

\begin{figure}
\centering
\includegraphics[scale=0.74]{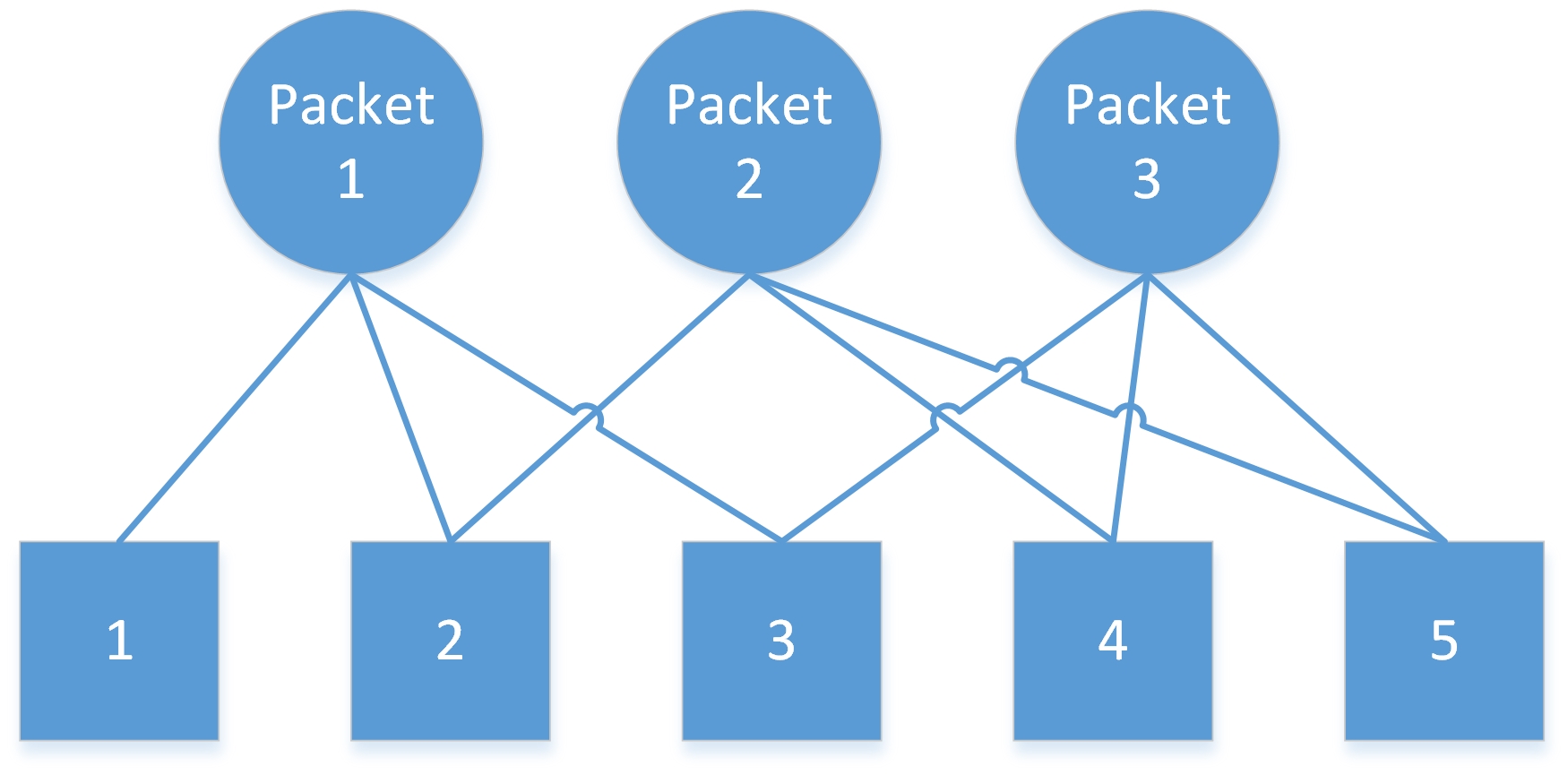}
\caption{nkMTP from Example \ref{ex1} formulated on a graph.}
\label{fig_ex1}
\end{figure}

%
%

\subsection{Complexity}
\label{sec:complexity}
In this subsection, we analyze the computational complexity of nkMTP. Clearly, a simple approach for solving nkMTP is to consider all possible assignments of MUs to packets, and to choose the assignment leading to the maximal number of packets that can be read. However, this approach is clearly inefficient since its complexity scales exponentially in $L$. In fact, nkMTP can be solved in polynomial time if $k=1, n \ge 1$ (i.e., each packet consists of one chunk and the repetition code is used) or $k=n=2$. On the other hand, nkMTP is NP-hard for $3 \le k \le n$.

\begin{theorem}
For $k=1, n \ge 1$ or $k=n=2$, nkMTP is solvable in polynomial time.
\label{th:k1poly}
\end{theorem}
\begin{proof}
When $k=1, n \ge 1$, nkMTP is equivalent to finding a subgraph $H$ of $G$ (a graph representation of nkMTP) that is a \textit{maximum bipartite matching} \cite{Bondy}, i.e., containing the largest number of \textit{matched} pairs $(x,y)$, $x \in X_G, y \in Y_G$, such that each pair is connected by an edge and the edges are pair-wise non-adjacent. When $k=n=2$, consider the $N$ MUs as the vertices of a (uni-partite) graph, where an edge in this graph connects two MUs shared by the same packet. A \textit{maximum matching} in this graph will provide the largest number of disjoint pairs of MUs, each pair serving a packet, corresponding to an optimal solution of the nkMTP instance. Efficient algorithms are known for finding maximum matching in both cases \cite{Bondy}. \qed
\end{proof}


%
%
%
%
%

\begin{theorem}
nkMTP is NP-hard for $3 \le k \le n$.
\label{np_hard} 
\end{theorem}
To prove Theorem \ref{np_hard}, we \textit{reduce} the $l$-set packing ($l$-SP) problem \cite{Hazan}, known to be NP-hard, to nkMTP. In $l$-SP, there are $L$ sets, each of size $l$, and the problem is to find the maximal number of pair-wise disjoint sets. The details of the reduction are provided in Appendix \ref{np_hard_detailed_proof}. The consequence of the hardness result of Theorem \ref{np_hard} is that no efficient algorithms are expected to be found for solving nkMTP when $3 \le k \le n$. However, in the next sections we provide algorithmic and analytic results that help solving the nkMTP problem in practical settings. We will also see variants of nkMTP for which we do find polynomial-time algorithms.

%

\section{Probabilistic Analysis and Bounds}
\label{sec:prob_analysis}
The fact that nkMTP turns out to be NP-hard for interesting coding parameters is important theoretically, but should not discourage one from seeking high-throughput coded switching. In this section we provide tools that will help find good coded-switching solutions in a practical setup. 

\subsection{Lower bounding the maximal solution}
In this sub-section, we provide a lower bound on the number of packets that can be read, for a given nkMTP instance. Consider the following randomized algorithm, applied to a graph formulation of nkMTP (Problem \ref{nkMTP_graph}):
\begin{algorithm}
\label{rand_algorithms}
\begin{enumerate}
\item Calculate the degree of each MU.
\item Assign each MU independently at random to one of its connected packets with probability ${1}/{d}$, where $d$ is the MU degree.
\item Return the set of packets having at least $k$ connected MUs.
\end{enumerate}
\end{algorithm}
The packets in the set returned by Algorithm \ref{rand_algorithms} can be read, since at least $k$ chunks are available for each packet. We now turn to calculate the expected size of the set returned by Algorithm \ref{rand_algorithms}, which will be shown to serve as a lower bound on $L^*$. Denote by $l_i$ an indicator random variable that equals $1$ if packet $i$ was assigned at least $k$ MUs, and equals $0$ otherwise. In addition, define auxiliary indicator random variables $l_{i,j}$, which equals $1$ if MU $j$ from the set of MUs connected to packet $i$ (i.e., $j \in S_i$, where $S_i$ are defined in Section \ref{sec:problem_formulation}) ends up connected to packet $i$. Thus:
\begin{equation}
\label{eq:larger_k}
\Pr \left( {{l_i} = 1} \right) = \Pr \left( {\sum\limits_{j \in S_i} {{l_{i,j}}}  \ge k} \right).
\end{equation}

Denote by $d_{i,j}$ the degree of MU $j$ connected to packet $i$. $l_{i,j}$ are Bernoulli random variables whose success probabilities are $1/d_{i,j}$. These random variables are independent for the same $i$ and different $j$, but not identically distributed (since the MU degrees may vary). The distribution of these random variables is called \textit{Poisson binomial distribution} \cite{Fernandez}. A closed-form expression for the right-hand side of Equation \eqref{eq:larger_k}, denoted $Q(i)$, is obtained using \cite{Fernandez}:
\begin{align}
\label{Qk}
Q\left( i \right) &= 1 - \sum\limits_{s = 0}^n {\left\{ {\left[ {\sum\limits_{t = 0}^{k - 1} {{e^{ - j2\pi st/\left( {n + 1} \right)}}} } \right]} \right.} \\ \nonumber
&\cdot \left. {\prod\limits_{j \in S_i} {\left[ {\left( {1/{d_{i,j}}} \right) \cdot {e^{j2\pi s/\left( {n + 1} \right)}} + \left( {1 - 1/{d_{i,j}}} \right)} \right]} } \right\}/\left( {n + 1} \right).
\end{align}
We are now ready to provide a lower bound on $L^*$.
\begin{theorem}
\label{theorem:prob}
For a given nkMTP instance, the number of packets that can be read is lower-bounded as follows:
\begin{equation}
\label{L_Q}
{L^*} \ge \sum\limits_{i = 1}^L {Q\left( i \right)}.
\end{equation}
\end{theorem}
\begin{proof}
Denote by $I$ the cardinality of the set returned by Algorithm \ref{rand_algorithms}. $I$ is a random variable and its expected value (over realizations of sets provided by Algorithm \ref{rand_algorithms}) is:
\begin{align}
\label{prob_method_lb}
E\left[ I \right] &= E\left[ {\sum\limits_{i = 1}^L {{l_i}} } \right] = \sum\limits_{i = 1}^L {E\left[ {{l_i}} \right]}  = \sum\limits_{i = 1}^L {\Pr \left( {{l_i} = 1} \right)}
 \\ \nonumber &= \sum\limits_{i = 1}^L {Q\left( i \right)}.
\end{align}
Since $E\left[ I \right]$ is an expected value, there must exist a valid solution to nkMTP with cardinality at least $E\left[ I \right]$. The existence of such set leads to the lower bound \eqref{prob_method_lb}. \qed
\end{proof}

\begin{example}
\label{ex_2}

Consider the nkMTP instance of Figure \ref{fig_ex1}. The MU degrees are ${d_{1,1}} = 1, {d_{1,2}} = 2, {d_{1,3}} = 2$ and so on. The lower bounds on $L^{*}$, obtained using \eqref{L_Q}, are $2.75$, $1.75$ and $0.5$ for $k=1$, $k=2$ and $k=3$, respectively, where the corresponding $L^*$ values are $3$, $2$ and $1$.
\end{example}


\subsection{Expected performance of nkMTP ensembles}
In this sub-section, we analyze the nkMTP in a random setting, where we consider random \textit{ensembles} of nkMTP instances with fixed parameters $k, n, N$ and $L$. Assuming a graph formulation (Problem \ref{nkMTP_graph}), an instance taken from an ensemble consists of a graph $G$ with $L$ packets and $N$ MUs, where the $n$ chunks of each packet are stored independently and uniformly at random at $n$ MUs. For each ensemble, we would like to estimate the probability of maximum throughput, i.e. the existence of a solution in form of a subgraph $H$ such that $X_H=X_G$. The method is to identify a condition for the existence of such a solution, derived from the following extension of Hall's theorem.
\begin{theorem}
\label{ex_hall}
\emph{(Extended Hall's Theorem~\cite{Viderman})}
Consider a bipartite graph $G = \left( {{X_G},{Y_G},{E_G}} \right)$. Then, $G$ satisfies $\deg(x)=k$ and $\deg (y) \in \left\{ {0,1} \right\}$ for all $x \in X_G$ and all $y \in Y_G$, if and only if for every subset $W$ of $X_G$, 
\begin{equation}
\label{PMCT}
\left| {T\left( W \right)} \right| \ge k\left| W \right|,
\end{equation}
where ${T}\left( W \right)$ is the set of vertices (in $Y_G$) adjacent to the vertices in $W$.
\end{theorem}
That is, the extended Hall's theorem provides a necessary and sufficient condition to determine whether MUs in $Y_G$ can be assigned to packets in $X_G$, such that the degree of each vertex in $X_G$ will be $k$. This is equivalent to the existence of a maximum throughput solution in $G$ (i.e., each packet can be served). In the sequel we say that \textit{Hall's condition} holds for a given subset $W$ of $X_G$ if the condition \eqref{PMCT} of Theorem \ref{ex_hall} holds. 


Denote by ${u_m} = {u_m}\left( {|W|,n;N} \right)$ the probability that the union of $|W|$ sets, each containing $n$ elements taken independently and uniformly at random from the set $\left\{ {1,2,...,N} \right\}$, results in a set of size $m$ ($m=1,2,...,N$). An expression for $u_m$ was derived in a previous work by us \cite{CC}, based on a Markov model that is an extension of the balls-and-bins model \cite{Mitz}. For a given subset $W$ of $X_G$, denote by ${P_W}$ the probability that the number of neighbours of $W$ is greater than or equal to $k\left| W \right|$ (meaning that Hall's condition holds for $W$). Since each vertex in $W$ is of degree $n$, $P_W$ is equivalent to the probability that the random union of $|W|$ sets of size $n$ (the number of MUs that store each packet) results in a set of size that is greater than or equal to $k|W|$. Thus,
\begin{equation}
\label{P_W}
P_W = \sum\limits_{m=k\left| W \right|}^N {{u_m}} \left( {\left| W \right|;n,N} \right).
\end{equation}
Note that $P_W$ from \eqref{P_W} is independent of the particular choice of $W$. Instead, it depends on the \textit{size} of $W$, i.e., ${P_W} = {P_{\left| W \right|}}$.

Calculating the probability that all subsets $W \subseteq X_G$ satisfy Hall's condition is a difficult task, due to dependencies between the subsets. However, a \textit{necessary} condition for an instance to contain a maximum throughput solution is that the set $X_G$ satisfies Hall's condition, which happens with probability $P_{|X_G|}$. Therefore, an upper bound on the probability that an instance contains a maximum throughput solution is simply $P_{|X_G|}$. As $P_{|X_G|}$ is obtained by raising a Markov matrix to power $|X_G|$, its calculation is efficient. In Figure \ref{fig:pw}, we first calculate the probability that a random graph $G$ contains a maximum throughput solution, by averaging over the optimal solution sizes of $10,000$ instances for $k=2$ and varying values of $n$. We then compare this empirical average to the upper bound obtained using $P_{|X_G|}$. The results show that the upper bound provided by $P_{|X_G|}$ is tight and captures the behaviour of the empirical average. The bound is tighter as $n$ becomes larger compared to $k$, meaning that in such cases the probability that Hall's condition is satisfied for all $X_G$ subsets is dominated by the probability that the set $X_G$ satisfies Hall's condition. As expected, better average performance is achieved as $n$ increases (i.e., when redundancy increases). It is also demonstrated that ${P_{\left| {{X_G}} \right|}} = 0$ for $\left| {{X_G}} \right| > \left\lfloor {N/k} \right\rfloor$, since the maximal number of packets that can be served is upper bounded by $\left\lfloor {N/k} \right\rfloor$.

Hall's condition provides a convenient way for estimating expected performance of an ensemble, when concentrating on maximum throughput solution. This way, instead of estimating the probability of maximum throughput empirically by averaging over the optimal solution sizes of many nkMTP instances (whose solution is hard in general, see Section \ref{sec:complexity}), we may resort to calculate $P_{|X_G|}$ instead. This is especially useful for large values of $n$ and $L$, for which a direct solution of nkMTP (e.g., by considering all possible assignments of MUs to packets), may be prohibitively complex. In addition, we obtained an efficient way for choosing parameters $k$ and $n$ and load $L$ such that good performance is expected.

\begin{figure}
\centering
\includegraphics[scale=0.62]{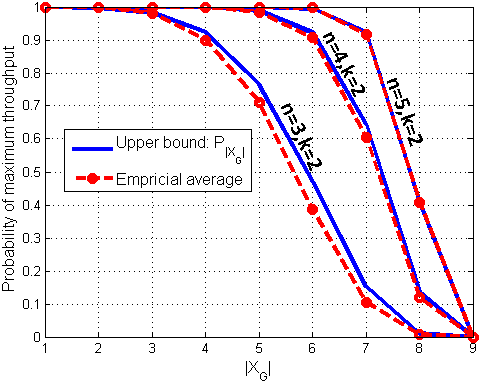}
\caption{The probability that all $|X_G|$ packets of a random nkMTP instance can be served, as a function of $k$ and $n$ ($N=16$).}
\label{fig:pw}
\end{figure}

\section{Polynomial-Time Solution for a Structured Variant of nkMTP}
\label{sec:special_cases}
Motivated by the hardness of the nkMTP problem proved above, we now turn to consider variants of the problem that can be solved efficiently. Our approach to make the problem easier to solve is by changing the packet writing policy upon packet arrival. We show in this section that, by a moderate restriction on the MUs chosen to store the $n$ chunks of a packet, we turn the maximal-throughput read problem to a tractable one. 

To the nkMTP problem discussed above we add the constraint that the $n$ encoded chunks of each packet are stored in $n$ \textit{consecutive} MUs. We will refer to this scheme as CnkMTP, for which we show that polynomial-time solution exists. We first sort the sets $S_i$ in a non-decreasing order of the maximal index of the MUs they contain, and w.l.o.g. we consider instances of CnkMTP in which the packets are sorted accordingly. We begin with an empty set of packets, denoted $\Lambda$, and an empty set of MUs, denoted $\Omega$. The following algorithm solves CnkMTP for $k \le n$.

\begin{algorithm}
\label{CnkMTP_algorithm}
\begin{enumerate}
\item Initialize $S'_i=S_i$ for $i=1,\ldots,L$.
\item Set $i:=1$.
\item If $|S'_i|\geq k$, add $i$ to $\Lambda$, and add the $k$ lowest elements of $S'_i$ to $\Omega$.
\item Remove all the elements added to $\Omega$ from all the sets $S'_j$, for $j>i$.
\item Set $i:=i+1$. If $i > L$, stop. Otherwise, go to step 3. 

\end{enumerate}
\end{algorithm}

\begin{theorem}
The set of packets $\Lambda$ and their corresponding MUs in $\Omega$ found by Algorithm \ref{CnkMTP_algorithm} are an optimal solution to CnkMTP.
\end{theorem}
\begin{proof}
The proof starts by observing that w.l.o.g the $k$ MUs assigned to a read packet are consecutive. If not, we can always exchange MUs between packet $j$ and packets $j'>j$ to make the assignment consecutive. 

Now we prove that if $|S'_i|\geq k$ in step 3, then $i$ appears in the optimal solution. We prove by induction on $i$. Assume all packets $1,\ldots,i-1$ can be chosen optimally according to Algorithm \ref{CnkMTP_algorithm}. Then we show that the $i$-th packet can be chosen in the same way. We assume by contradiction that $|S'_i|\geq k$ and there is no optimal solution that contains packet $i$. Then we look at the smallest $j>i$ for which packet $j$ appears in the optimal solution. Then from the fact that its $k$ assigned MUs are consecutive it is possible to shift the assignment to the first MU index in $S'_i$, and replace $j$ by $i$ in the optimal solution without affecting the selection of any $j'>j$. This is a contradiction.\qed
\end{proof}

The operations required in Algorithm \ref{CnkMTP_algorithm} are simple shifting, in addition to the sorting of the packets. Hence its running time is clearly polynomial. We also consider a special case of CnkMTP, in which the $N$ MUs are partitioned into ${N}/{n}$ blocks of size $n$ (assuming that $n$ is a divisor of $N$), where the encoded chunks of each packet are restricted to one of those blocks. We term this scheme as CnkMTP \textit{on blocks}. A comparison between nkMTP and CnkMTP is given in Figure \ref{fig:scheme_comp} for $k=3, n=4$ and $N=16$. This figure reveals a very interesting tradeoff between throughput and computational complexity. Adding structure to the write policy in CnkMTP results in some loss of throughput compared with the unrestricted nkMTP. However, this throughput can be attained efficiently, where for nkMTP it is intractable to reach the throughput efficiently for large problem instances. It is also shown that CnkMTP gives better throughput compared with an even more structured writing policy, CnkMTP on blocks.

\begin{figure}
\centering
\includegraphics[scale=0.62]{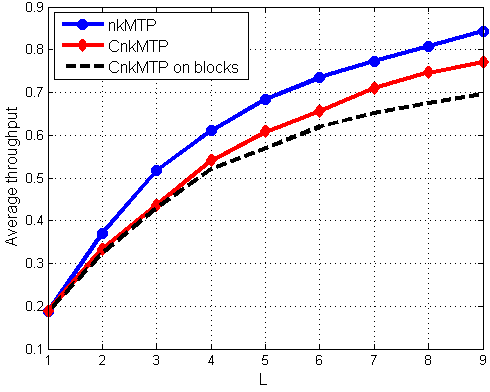}
\caption{Scheme comparison.}
\label{fig:scheme_comp}
\end{figure}

\section{Conclusion}
\label{sec:conclusions}

In this paper, we anaylzed the fundamental limits of using MDS codes in a switching environment. We proved that in its most general form, the problem of obtaining maximum throughput for a set of requested packets is a hard problem. Therefore, we provided bounds and algorithmic tools to aid its solution in practice. By a simple modification of the writing policy used by the switch, we have shown how the problem can be solved efficiently. Our work leaves many interesting problems for future research, most immediately how to tailor the switch and code parameters to match real-life network workloads. From a practical point of view, we currently investigate the performance of certain MDS codes in a switching environment, taking encoding/decoding overhead into account.

\section{Acknowledgement}

The authors would like to thank A. Nekrasov and S. Shulga for their contribution to performance evaluation simulations. This work was partly conducted under a joint ISF-UGC grant. In addition, it was supported by a Marie Curie CIG grant and by the Israel Ministry of Science and Technology.

\bibliographystyle{IEEEtran}

\bibliography{MDS_paper}

\appendices

\section{Detailed proof of Theorem \ref{np_hard}}
\label{np_hard_detailed_proof}
To show hardness of nkMTP when $3 \le k \le n$, we first define the decision-problem version of nkMTP, which we name $M$-nkMTP. In the rest of this appendix, we assume that $3 \le k \le n$.

\begin{problem}{($M$-nkMTP)}
\label{nkMTP_decision}

\textbf{Input}: Set theory formulation (Problem \ref{nkMTP_set}) of nkMTP and a positive integer $M$.

\textbf{Output}: "Yes" if there are $M$ subsets $S'_i \subseteq S_i$ with the properties $\left| S'_i \right| = k$, $S'_i \cap S'_j = \emptyset$ ($i \ne j$).
\end{problem}
For showing that nkMTP is NP-hard we can equivalently show that $M$-nkMTP is NP-complete. Note that $M$-nkMTP is in NP, since once we are given a collection of $M$ subsets $S'_i \subseteq S_i$ claimed to be pair-wise disjoint, this can validated in polynomial time. It remains to reduce a known NP-complete problem to $M$-nkMTP, meaning that we have to show that an efficient solution to $M$-nkMTP implies an efficient solution to this NP-complete problem. We will reduce the \textit{$l$-set packing} problem ($l$-SP), known to be NP-complete for $l \ge 3$  \cite{Hazan}, to our problem. $l$-SP is defined as follows.
\begin{problem}{($l$-SP)}
\label{l_SP}

\textbf{Input}: Collection of sets over a certain domain, each of them of size $l$, and a positive integer $M$.

\textbf{Output}: "Yes" if there are $M$ pair-wise disjoint sets.
\end{problem}
\begin{proof}
First, $M$-nkMTP is NP-complete for $3 \le k = n$, since in this case $M$-nkMTP and $l$-SP, for $l=k=n$, are essentially the same. Therefore, it remains to reduce $l$-SP ($l \ge 3$) to $M$-nkMTP for $3 \le k < n$. Let us begin with reducing $l$-SP to $M$-nkMTP with $k=l, n = k+1$.

Consider an instance of $l$-SP with $l=k$, with $M$ denoting the number of pair-wise disjoint subsets required in the solution. Assume that the input to $l$-SP are $L$ sets $A_i$ ($i=1,2,...,L$), where the elements contained in $A_i$ are $\bigcup\limits_i {{A_i}}  = \left\{ {{a_1},{a_2},...,{a_s}} \right\}$. For building an instance of $M$-nkMTP with $k=l, n=k+1$, do the following:

\begin{itemize}
\item Build sets $B_i$, each of size $k$, from $s$ new elements $\left\{ {{b_1},{b_2},...,{b_s}} \right\}$, such that a one-to-one correspondence between the elements in $A_i$ and the elements in $B_i$ exists: ${a_j} \in {A_i} \Leftrightarrow {b_j} \in {B_i}$.
\item Add a new element, say $\theta$, which does not appear in the sets $A_i$ or $B_i$, to each of the sets $A_i$ and $B_i$. Denote the new sets by ${{\tilde A}_i}$ and ${{\tilde B}_i}$.
\end{itemize}

The input to $M$-nkMTP with $n=k+1$ will be the sets ${{\tilde A}_i}$ and ${{\tilde B}_i}$, where we will ask whether there exist $2M$ subsets of size $k$ each that are pair-wise disjoint. If $l$-SP provides a solution of size $M$ for the sets $A_i$, then clearly the sets ${A_i} \subseteq {\tilde A_i}, {B_i} \subseteq {\tilde B_i}$ serve as solution of size $2M$ to $M$-nkMTP with $n=k+1$. On the other hand, if there exists a solution of size $2M$ in the $M$-nkMTP problem, we have three cases:
\begin{enumerate}
\item $M$ subsets $A{'_i} \subseteq {{\tilde A}_i}$ and $M$ subsets $B{'_i} \subseteq {{\tilde B}_i}$ appear in the solution. The element $\theta$ can appear in only one of the subsets, since they must be pair-wise disjoint. If $\theta$ belongs to some $A{'_i}$, then we have $M$ subsets $B{'_i}$ that provide a solution to $l$-SP (after transforming the elements in $B{'_i}$ to the their corresponding elements in $A{'_i}$). On the other hand, if $\theta$ belongs to some $B{'_i}$, then the solution is the sets $A{'_i}$. 
\item $M_1$ subsets $A{'_i} \subseteq {{\tilde A}_i}$ and $M_2$ subsets $B{'_i} \subseteq {{\tilde B}_i}$ appear in the solution, where $M_1 < M_2$ and $M_1 + M_2 = 2M$. $\theta$ can appear in at most one of the subsets $B'_i$. In addition, $M < M_2$, and therefore choosing the subsets $B'_i$ that do not contain $\theta$ leads to a solution of $l$-SP with at least $M$ subsets (again, transformation to the elements of $A_i$ is required).
\item $M_1$ subsets $B{'_i} \subseteq {{\tilde B}_i}$ and $M_2$ subsets $A{'_i} \subseteq {{\tilde A}_i}$ appear in the solution, where $M_1 < M_2$ and $M_1 + M_2 = 2M$. A solution of size at least $M$ to $l$-SP is obtained in a similar way to the previous case.
\end{enumerate}

The transformation ${A_i} \to {\tilde A_i},{\tilde B_i}$ is clearly polynomial in $L$, since it merely requires to build $L$ sets of size $k$ and to add one element to each of the resulting $2L$ sets. Thus, the reduction described above is a polynomial time reduction. Therefore, $M$-nkMTP is NP-complete for $k \ge 3, n = k+1$, and it remains to show that $M$-nkMTP is NP-complete for $k \ge 3, n > k+1$.

Consider $M$-nkMTP with $k \ge 3, n = k+2$. We can reduce $M$-nkMTP with $k \ge 3, n = k+1$ (which we proved to be NP-complete) to $M$-nkMTP with $k \ge 3, n = k+2$, similarly to the reduction of $l$-SP to $M$-nkMTP with $k=l, n=k+1$ that was described earlier. Continuing in the same fashion, we are able to reduce $M$-nkMTP with $n=k+j$ ($k \ge 3, j \ge 1$) to $M$-nkMTP with $n=k+j+1$. Finally, we deduce that $M$-nkMTP is NP-complete for $3 \le k \le n$, meaning that nkMTP (the optimization version of $M$-nkMTP) is NP-hard. \qed
\end{proof}

\end{document}